 \documentclass[preprint,12pt,times,sort&compress]{elsarticle}

\usepackage{amssymb}
\usepackage{amsthm}
\theoremstyle{plain}
\newtheorem{theorem}{Theorem}
\newtheorem{lemma}[theorem]{Lemma}

\theoremstyle{definition}
\newtheorem{definition}[theorem]{Definition}
\newtheorem{algo}[theorem]{Algorithm}

\usepackage{makeidx}  
\usepackage[fleqn]{amsmath}
\usepackage{hyperref}

\journal{arXiv}

\begin{document}




\title{A recognition algorithm for adjusted interval digraphs}


\author[KU]{Asahi~Takaoka}\ead{takaoka@kanagawa-u.ac.jp}
\address[KU]{
  Department of Information Systems Creation, 
  Faculty of Engineering, 
  Kanagawa University, \\ 
  Rokkakubashi 3-27-1, Kanagawa-ku, 
  Yokohama-shi, 
  Kanagawa, 221--8686, Japan 
}


\begin{abstract}
Min orderings give a vertex ordering characterization, 
common to some graphs and digraphs 
such as interval graphs, 
complements of threshold tolerance graphs 
(known as co-TT graphs), and 
two-directional orthogonal ray graphs. 
An adjusted interval digraph is 
a reflexive digraph that has a min ordering. 
Adjusted interval digraph can be recognized in $O(n^4)$ time, 
where $n$ is the number of vertices of the given graph. 
Finding a more efficient algorithm is posed 
as an open question. 
This note provides a new recognition algorithm 
with running time $O(n^3)$. 
The algorithm produces a min ordering 
if the given graph is an adjusted interval digraph. 

\end{abstract}


\maketitle


\section{Introduction}
All graphs and directed graphs (digraphs for short) 
considered in this paper are 
finite and have no multiple edges 
but may have loops. 
We write $uv$ for the undirected edge 
joining a vertex $u$ and a vertex $v$; 
we write $(u, v)$ for the directed edge 
from $u$ to $v$. 
We write $V(H)$ for the vertex set of a digraph $H$; 
we write $E(H)$ for the edge set of $H$. 
We say that $u$ \emph{dominates} $v$ 
(and that $v$ \emph{is dominated by} $u$) 
in a digraph $H$ if $(u, v) \in E(H)$, 
and denote it by $u \to v$ or $v \gets u$. 

A digraph $H$ is an \emph{interval digraph}~\cite{SDRW89-JGT} if 
for each vertex $v$ of $H$, 
there is a pair of intervals $I_v$ and $J_v$ 
on the real line such that 
$u \to v$ in $H$ 
if and only if $I_u$ intersects $J_v$. 
An interval digraph is 
an \emph{adjusted interval digraph}~\cite{FHHR12-DAM} 
if the two intervals $I_v$ and $J_v$ have the same left endpoint 
for each vertex $v$. 
A digraph is called \emph{reflexive} if every vertex has a loop, 
and every adjusted interval digraph is reflexive by definition. 

Adjusted interval digraphs have been introduced 
by Feder et al.~\cite{FHHR12-DAM} in connection with 
the study of list homomorphisms. 
They have shown two characterizations 
and a recognition algorithm of this graph class. 

A \emph{min ordering} of a digraph $H$ is a linear ordering $\prec$ 
of the vertices of $H$ such that 
for any two edges $(u, v)$ and $(u', v')$ of $H$, 
we have $(u, v') \in E(H)$ 
if $u \prec u'$ and $v' \prec v$. 
We remark that $(u, v)$ can be a loop, 
and similarly for $(u', v')$. 
A reflexive digraph has a min ordering 
if and only if it is an adjusted interval digraph~\cite{FHHR12-DAM}. 
Min orderings give similar characterizations of some graph classes 
such as interval graphs, 
complements of threshold tolerance graphs 
(known as co-TT graphs)~\cite{MRT88-JGT}, 
two-directional orthogonal ray graphs~\cite{STU10-DAM}, and 
signed-interval digraphs~\cite{HHMR18-LIPIcs}. 
See~\cite{HHMR18-LIPIcs} for details. 

Suppose that a digraph $H$ has a min ordering $\prec$, and 
let $(u, v), (u', v')$ be two edges of $H$ 
with $(u, v') \notin E(H)$. 
We have $v \neq v'$ from 
$(u, v) \in E(H)$ and $(u, v') \notin E(H)$; 
similarly, we have $u \neq u'$ from 
$(u', v') \in E(H)$ and $(u, v') \notin E(H)$. 
If $u \prec u'$ and $v' \prec v$, then 
$(u, v') \in E(H)$ by the property of min orderings, 
a contradiction. 
Thus, 
$u \prec u'$ implies $v \prec v'$ and 
$v' \prec v$ implies $u' \prec u$. 
We can capture 
this forcing relation 
with an auxiliary digraph. 
The \emph{pair digraph} $H^+$ 
associated with a digraph $H$ is a digraph 
such that the vertex set $V(H^+)$ is the set 
$\{(u, v) \colon\ u \neq v\}$ 
of ordered pair of two vertices of $H$, and 
$(u, u') \to (v, v')$ and 
$(v', v) \to (u', u)$ in $H^+$ 
if and only if 
$(u, v), (u', v') \in E(H)$ and $(u, v') \notin E(H)$. 

An \emph{invertible pair} of a digraph $H$ is 
a pair of two vertices $u, v$ of $H$ such that 
in $H^+$, the vertices $(u, v)$ and $(v, u)$ are 
in the same strong component. 
It is clear that if $H$ has an invertible pair, 
then $H$ does not have any min ordering. 
Feder et al.~\cite{FHHR12-DAM} have shown that 
the converse also holds; therefore, 
a reflexive digraph has no invertible pairs 
if and only if it has a min ordering. 

The characterizations of adjusted interval digraphs yield 
a recognition algorithm with running time $O(m^2 + n^2)$, 
where $n$ and $m$ are the number of vertices and edges 
of the given graph, respectively~\cite{FHHR12-DAM}. 
Finding a linear-time recognition algorithm is posed as 
an open question~\cite{FHHR12-DAM,HHMR18-LIPIcs}. 
In this paper, we show an $O(n^3)$-time recognition algorithm 
for adjusted interval digraphs. 
The algorithm produces a min ordering or 
finds an invertible pair of the given graph if it exists. 
As a byproduct, we also give an alternative proof to show that 
a reflexive digraph has a min ordering 
if and only if it has no invertible pairs.

\section{Algorithm}
In the case of reflexive digraphs, 
there is an equivalent simpler definition of min orderings. 
\begin{theorem}[Feder et al.~\cite{FHHR12-DAM}]
Let $H$ be a reflexive digraph. 
A linear ordering $\prec$ of the vertices of $H$ is a min ordering 
if and only if for any three vertices $u, v, w$ 
with $u \prec v \prec w$, 
\begin{enumerate}[\bfseries --]
\item $(u, w) \in E(H)$ implies $(u, v) \in E(H)$, and 
\item $(w, u) \in E(H)$ implies $(v, u) \in E(H)$. 
\end{enumerate}
\end{theorem}
In other words, 
a linear ordering $\prec$ of the vertices 
of a reflexive digraph is a min ordering if 
it contains no triples of vertices $u, v, w$ with 
$u \prec v \prec w$ such that 
$(u, w) \in E(H)$ and $(u, v) \notin E(H)$, or
$(w, u) \in E(H)$ and $(v, u) \notin E(H)$. 
We call such triples of vertices the \emph{forbidden patterns}. 

Let $H$ be an adjusted interval digraph 
with a min ordering $\prec$. 
Let $u, v, w$ be distinct vertices of $H$ with 
$(u, w) \in E(H)$ and $(u, v) \notin E(H)$, or 
$(w, u) \in E(H)$ and $(v, u) \notin E(H)$. 
In both cases, if $u \prec v \prec w$ 
then we have a forbidden pattern. 
Thus, 
$u \prec v$ implies $w \prec v$ and $v \prec w$ implies $v \prec u$. 
To capture this forcing relation, 
we define an auxiliary digraph associated with $H$. 
\begin{definition}
Let $H$ be a reflexive digraph. 
The \emph{implication graph} $H^*$ of $H$ 
is a digraph such that 
the vertex set $V(H^*)$ is the set 
$\{(u, v) \colon\ u \neq v\}$ 
of ordered pair of two vertices of $H$, and 
for any three vertices $u, v, w$ of $H$, 
we have $(u, v) \to (w, v)$ and $(v, w) \to (v, u)$ in $H^*$ 
if and only if 
\begin{enumerate}[\bfseries --]
\item $(u, w) \in E(H)$ and $(u, v) \notin E(H)$, or 
\item $(w, u) \in E(H)$ and $(v, u) \notin E(H)$. 
\end{enumerate}
\end{definition}

We can use the implication graphs 
for recognizing adjusted interval digraphs. 
\begin{lemma}\label{lemma:implication-digraph}
Let $H$ and $H^*$ be a reflexive digraph 
and its implication graph, respectively. 
A pair of two vertices $u, v \in V(H)$ is 
an invertible pair if and only if in $H^*$, 
the vertices $(u, v)$ and $(v, u)$ are 
in the same strong component. 
\end{lemma}
\begin{proof}
Let $H^+$ be the pair digraph of $H$. 
Let $u, v, w$ be three vertices of $H$ such that 
$(u, v) \to (w, v)$ in $H^*$ 
(or equivalently, $(v, w) \to (v, u)$ in $H^*$). 
By definition, 
$(u, w) \in E(H)$ and $(u, v) \notin E(H)$, or 
$(w, u) \in E(H)$ and $(v, u) \notin E(H)$. 
Since the vertex $v$ has a loop, 
in both cases 
$(u, v) \to (w, v)$ and $(v, w) \to (v, u)$ in $H^+$. 
Therefore, $H^*$ is a subgraph of $H^+$. 
\par
Assume that $(u, v) \to (u', v')$ in $H^+$. 
By definition, 
$(u, u'), (v, v') \in E(H)$ and $(u, v') \notin E(H)$, or 
$(u', u), (v', v) \in E(H)$ and $(v', u) \notin E(H)$. 
In both cases, if $(u, u')$ or $(v, v')$ is a loop, 
then $(u, v) \to (u', v')$ in $H^*$. 
Thus we may assume $u \neq u'$ and $v \neq v'$. 
We have $u \neq v'$ since $H$ is reflexive. 
Recall that $u \neq v$ and $u' \neq v'$. 
Thus, the vertices $u, v, v'$ are distinct, 
and $(u, v) \to (u, v')$ in $H^*$. 
Similarly, the vertices $u, u', v'$ are distinct, 
and $(u, v') \to (u', v')$ in $H^*$. 
Therefore, if $(u, v) \to (u', v')$ in $H^+$, then 
$H^*$ has a directed path from $(u, v)$ to $(u', v')$. 
\end{proof}
Lemma~\ref{lemma:implication-digraph} gives an algorithm 
to find an invertible pair if it exists. 
Given a reflexive digraph $H$, 
the algorithm 
first construct the implication graph $H^*$ of $H$, 
then compute the strong components of $H^*$, and 
finally check for the existence of a pair $(u, v)$ 
and $(v, u)$ within one strong component. 
The implication graph $H^*$ has $n(n-1)$ vertices, and 
at most $2nm$ edges 
since $H^*$ has at most two edges 
for each pair of a vertex and an edge of $H$. 
Therefore, we can construct $H^*$ 
in time $O(nm)$, and 
check for the existence of invertible pairs 
in the same time bound.

We next describe the algorithm for producing 
a min ordering of an adjusted interval digraph. 
Let $H$ and $H^*$ be a reflexive digraph and 
its implication graph, respectively. 
As an auxiliary graph, 
we use a complete graph $K$ with the vertex set $V(H)$. 
An \emph{orientation} of $K$ is 
a digraph obtained from $K$ 
by orienting each edge of $K$, that is, 
replacing each edge $uv \in E(K)$ 
with either $(u, v)$ or $(v, u)$. 
An orientation of $K$ is \emph{acyclic} 
if it contains no directed cycles; 
an acyclic orientation of $K$ 
is equivalent to a linear ordering of the vertices of $H$. 

We say that a vertex $(u, v)$ of $H^*$ 
is an \emph{implicant} of a vertex $(u', v')$ 
if $H^*$ has a directed walk 
from $(u', v')$ to $(u, v)$. 
We say that 
an orientation $T$ of $K$ 
\emph{is consistent with} $H$ if 
for each vertex $(u, v)$ of $H^*$, we have 
$u \to v$ in $T$ implies $u' \to v'$ 
for every implicant $(u', v')$ of $(u, v)$. 
It is clear that 
an acyclic orientation of $K$ is consistent with $H$ 
if and only if it contains no forbidden patterns 
of min orderings. 
Therefore, a min ordering of $H$ is equivalent to 
an orientation of $K$ that is acyclic and 
consistent with $H$. 

It is sufficient for the existence of a min ordering of $H$ 
that there is an orientation of $K$ consistent with $H$. 
\begin{lemma}\label{lemma:propagation}
There is a min ordering of $H$ if and only if 
there is an orientation of $K$ consistent with $H$. 
\end{lemma}
Let $T$ be an orientation of $K$ consistent with $H$ 
that is not acyclic. 
In order to prove Lemma~\ref{lemma:propagation}, we provide 
an algorithm for producing 
another orientation of $K$ that is acyclic and 
consistent with $H$. 

A \emph{directed triangle} is a directed cycle of length 3. 
It is well known that 
an orientation of a complete graph is acyclic if and only if 
it contains no directed triangles. 
Let $u$ be a vertex of $K$, and 
let $E_u$ be the set of all the edges $(v, w) \in E(T)$ 
such that $u \to v$, $v \to w$, and $w \to u$ in $T$. 
The \emph{reversal} $E_u^-$ of $E_u$ is 
the set of directed edges obtained from $E_u$ 
by reversing the direction of all the edges in $E_u$, 
that is, $E_u^- = \{(x, y) \colon\ (y, x) \in E_u\}$. 
We define that $T'$ is 
the orientation of $K$ obtained from $T$ 
by reversing the direction of all the edges in $E_u$, 
that is, $E(T') = (E(T) - E_u) \cup E_u^-$. 

We will show that the orientation $T'$ 
has the following properties: 
$T'$ is still consistent with $H$; 
$T'$ contains no directed triangles having the vertex $u$; 
the reversing the direction of edges in $E_u$ generates 
no directed triangles. 
Therefore, by repeated application of 
this procedure for each vertex of $K$, 
we can obtain an orientation of $K$ that is acyclic and 
consistent with $H$; 
the complexity of the algorithm is $O(n^3)$. 

To show that $T'$ is still consistent with $H$, 
we prove a lemma for directed triangles 
in the orientation of $K$ consistent with $H$. 
\begin{lemma}\label{lemma:triangle}
Let $T$ be an orientation of $K$ 
consistent with $H$. 
Suppose that $T$ has three vertices 
$u, v, w$ such that 
$u \to v$, $v \to w$, and $w \to u$ in $T$. 
If $v' \to w'$ in $T$ and 
$(v', w') \to (v, w)$ in $H^*$, then 
$u \to v'$, $v' \to w'$, and $w' \to u$ in $T$. 
\end{lemma}
\begin{proof}
We say that 
a set of vertices $S \subseteq V(H)$ is 
\emph{complete} if $(x, y), (y, x) \in E(H)$ 
for any two vertices $x, y \in S$, and 
is \emph{independent} 
if $(x, y), (y, x) \notin E(H)$. 
We claim that the set of vertices 
$\{u, v, w\}$ is complete or independent. 
Suppose $(u, v) \in E(H)$. 
If $(w, v) \notin E(H)$, then 
$(v, w) \to (u, w)$ in $H^*$, a contradiction. 
Thus $(w, v) \in E(H)$. 
If $(w, u) \notin E(H)$, then 
$(w, u) \to (v, u)$ in $H^*$, a contradiction. 
Thus $(w, u) \in E(H)$. 
If $(v, u) \notin E(H)$, then 
$(u, v) \to (w, v)$ in $H^*$, a contradiction. 
Thus $(v, u) \in E(H)$. 
Continuing in this way, we have that 
$\{u, v, w\}$ is complete. 
\par
We have either $v' = v$ or $w' = w$. 
Suppose $v' = v$. 
Since $(v', w') \to (v, w)$ in $H^*$, we have 
$(v, w) \notin E(H)$ and $(w', w) \in E(H)$, or 
$(w, v) \notin E(H)$ and $(w, w') \in E(H)$. 
In both cases, 
the set of vertices $\{u, v, w\}$ is independent. 
Since $(u, w) \notin E(H)$ and $(w', w) \in E(H)$, or 
$(w, u) \notin E(H)$ and $(w, w') \in E(H)$, 
we have $(w, u) \to (w', u)$ in $H^*$; 
therefore, $w' \to u$ in $T$. 
\par
We next suppose $w' = w$. 
Since $(v', w') \to (v, w)$ in $H^*$, we have 
$(v', w') \notin E(H)$ and $(v', v) \in E(H)$, or 
$(w', v') \notin E(H)$ and $(v, v') \in E(H)$. 
Due to symmetry, we may assume 
$(v', w') \notin E(H)$ and $(v', u) \in E(H)$. 
If $(v', u) \in E(H)$, we have 
$(v', w') \to (u, w)$ in $H^*$, a contradiction. 
Thus $(v', u) \notin E(H)$. 
We now have $(u, v) \to (u, v')$ in $H^*$, and 
therefore, $u \to v'$ in $T$. 
\end{proof}

Suppose that $T'$ is not consistent with $H$. 
Then, there exist three vertices $x, y, z$ such that 
$x \to y$ and $y \to z$ in $T'$ but $(x, y) \to (z, y)$ in $H^*$ 
(or equivalently, $(y, z) \to (y, x)$ in $H^*$). 
Since $T$ is consistent with $H$, we have 
$(x, y) \in E_u^-$ or $(y, z) \in E_u^-$. 
Suppose $(x, y), (y, z) \in E_u^-$. 
We have that $(x, y) \in E_u^-$ implies $u \to y$ in $T$ 
and $(y, z) \in E_u^-$ implies $y \to u$ in $T$, a contradiction. 
If $(x, y) \in E_u^-$ and $(y, z) \notin E_u^-$, 
then $(y, x) \in E_u$ and $y \to z$ in $T$. 
Since $(y, z) \to (y, x)$ in $H^*$, 
we have from Lemma~\ref{lemma:triangle} that 
$(y, z) \in E_u$, a contradiction. 
Similarly, if $(x, y) \notin E_u^-$ and $(y, z) \in E_u^-$, 
then $x \to y$ in $T$ and $(z, y) \in E_u$. 
Since $(x, y) \to (z, y)$ in $H^*$, 
we have from Lemma~\ref{lemma:triangle} that 
$(x, y) \in E_u$, a contradiction. 
Therefore, $T'$ is still consistent with $H$. 

Trivially, $T'$ contains no directed triangles having the vertex $u$. 
\par
Let $x, y, z$ be three vertices 
such that $x \to y$, $y \to z$, and $z \to x$ in $T'$. 
Suppose $(x, y), (y, z) \in E_u^-$. 
We have that $(x, y) \in E_u^-$ implies $u \to y$ in $T$ and 
$(y, z) \in E_u^-$ implies $y \to u$ in $T$, a contradiction. 
Thus at most one edge on the directed triangle is in $E_u^-$. 
Suppose $(x, y) \in E_u^-$ and $(y, z), (z, x) \notin E_u^-$. 
We have $u \to y$, $y \to z$, $z \to x$, and $x \to u$ in $T$. 
If $u \to z$ in $T$ then $(z, x) \in E_u$; 
if $z \to u$ in $T$ then $(y, z) \in E_u$, a contradiction. 
Therefore, the reversing the direction of edges in $E_u$ generates 
no directed triangles, 
and we have Lemma~\ref{lemma:propagation}.

We now show an algorithm for finding 
an orientation of $K$ consistent with $H$. 
We use an algorithm for the 2-satisfiability problem. 
An instance of the 2-satisfiability problem is 
a \emph{2CNF formula}, 
a Boolean formula in conjunctive normal form 
with at most two literals per clause. 
We construct the 2CNF formula $\phi_H$ associated with $H$. 
Assume that the vertices of $H$ are linearly ordered, and 
let $x_{(u, v)}$ be a Boolean variable 
if a vertex $u$ of $H$ precedes a vertex $v$ 
in this ordering. 
We denote the negation of $x_{(u, v)}$ by $x_{(v, u)}$. 
We define that $\phi_H$ is a 2CNF formula 
consisting of all the clauses 
$(x_{(u, v)} \vee x_{(v, w)})$ such that 
$(u, w) \in E(H)$ and $(u, v) \notin E(H)$, or 
$(w, u) \in E(H)$ and $(v, u) \notin E(H)$. 
\par
Let $\tau$ be a truth assignment of $\phi_H$. 
We define that an orientation of $K$ associated with $\tau$ 
is an orientation $T_{\tau}$ such that 
$u \to v$ in $T_{\tau}$ if and only if $x_{(u, v)} = 0$ in $\tau$ 
for any two vertices $u, v$ of $K$. 
It is clear from the construction of $\phi_H$ that 
$T_{\tau}$ is consistent with $H$ 
if and only if $\tau$ satisfies $\phi_H$. 
\begin{lemma}\label{lemma:2SAT}
There is an orientation of $K$ consistent with $H$ 
if and only if $\phi_H$ is satisfiable. 
\end{lemma}

The 2CNF formula $\phi_H$ has at most $n(n-1)/2$ Boolean variables, and 
at most $nm$ clauses since $\phi_H$ has at most one clause 
for each pair of a vertex and an edge of $H$. 
Thus $\phi_H$ can be constructed in $O(nm)$ time. 
Since a satisfying truth assignment of $\phi_H$ can be computed 
in time linear to the size of $\phi_H$ (see~\cite{APT79-IPL} for example), 
we can find an orientation of $K$ consistent with $H$ 
in $O(nm)$ time. 

Let $\phi$ be a 2CNF formula. 
For a Boolean variable $x_i$ in $\phi$, 
the negation of $x_i$ is denoted by $\overline{x_i}$. 
The \emph{implication graph} $G(\phi)$ of $\phi$ 
is the digraph constructed as follows: 
for each variable $x_i$, 
we add two vertices named 
$x_i$ and $\overline{x_i}$ to $G(\phi)$; 
for each clause $(x_i, x_j)$, 
we add two edges to $G(\phi)$ so that 
$\overline{x_i} \to x_j$ and $\overline{x_j} \to x_i$. 
A 2CNF formula $\phi$ is satisfiable 
if and only if in $G(\phi)$, 
any pair of vertices $x_i$ and $\overline{x_i}$ are not 
in the same strong component~\cite{APT79-IPL}. 

For a reflexive digraph $H$, 
it is clear from the construction of $\phi_H$ and $H^*$ that 
$G(\phi_H)$ is isomorphic to the subgraph of $H^*$ 
obtained by removing all the isolated vertices of $H^*$. 
Therefore, we have the following. 
\begin{lemma}\label{lemma:APT79}
The 2CNF formula $\phi_H$ is satisfiable 
if and only if $H$ has no invertible pairs. 
\end{lemma}

From Lemmas~\ref{lemma:propagation},~\ref{lemma:2SAT}, 
and~\ref{lemma:APT79}, we now have an alternative proof 
for the theorem of Feder et al.~\cite{FHHR12-DAM}. 
\begin{theorem}
A reflexive digraph has a min ordering 
if and only if it contains no invertible pairs. 
\end{theorem}

We finally summarize our algorithm 
for recognizing adjusted interval graphs. 
This algorithm produces a min ordering of the given graph 
if it is an adjusted interval digraph, 
and finds an invertible pair if otherwise. 
\begin{algo}\label{algorithm}
Let $H$ be a reflexive digraph. \\
\begin{tabular}{ll}
  \emph{Step~1.}&
  Compute a 2CNF formula $\phi_H$ from $H$. 
  \\
  \emph{Step~2.}&
  Find a satisfying truth assignment of $\phi_H$. 
  \\
  & 
  If $\phi_H$ is satisfiable, go to Step~3. Otherwise, go to Step~4. 
  \\
  \emph{Step~3.}&
  Let $\tau$ be a satisfying truth assignment of $\phi_H$. 
  \\
  &
  Compute an orientation $T_{\tau}$ of $K$ associated with $\tau$. 
  \\
  &
  Compute a min ordering of $H$ from $T_{\tau}$ if $T_{\tau}$ is not acyclic. 
  \\
  &
  Output the min ordering of $H$, and halt. 
  \\
  \emph{Step~4.}&
  Construct the implication graph $H^*$ of $H$. Then, find an invertible pair. 
  \\
  &
  Output the invertible pair of $H$, and halt. 
\end{tabular}
\end{algo}
The correctness of the algorithm follows 
from Lemmas~\ref{lemma:propagation},~\ref{lemma:2SAT}, 
and~\ref{lemma:APT79}. 
Steps~1,~2, and~4 can be performed in $O(nm)$ time; 
Step~3 can be performed in $O(n^3)$ time. 
\begin{theorem}
Adjusted interval digraphs can be recognized in $O(n^3)$ time. 
\end{theorem}











\end{document}